\documentclass[journal]{IEEEtran}

\usepackage{graphicx,url}
\usepackage{dcolumn}
\usepackage{bm}
\usepackage{amsmath,amsfonts,amssymb,citesort}

\newtheorem{definition}{\noindent{\it Definition}}[section]
\newtheorem{theorem}{\noindent{\it Theorem}}[section]

\newtheorem{lemma}[theorem]{\noindent{\it Lemma}}

\newtheorem{remark}[theorem]{\noindent{\it Remark}}
\newtheorem{corollary}[theorem]{\noindent{\it Corollary}}
\newenvironment{proof}{\noindent{\it Proof:}}{$\hfill$ $\Box$\\ }
\newtheorem{example}{\noindent{\it Example}}[section]

\begin{document}

\title{On Classical and Quantum MDS-Convolutional BCH Codes}

 \author{Giuliano G. La Guardia
\thanks{Giuliano Gadioli La Guardia is with Department of Mathematics and Statistics,
State University of Ponta Grossa (UEPG), 84030-900, Ponta Grossa, PR, Brazil.
}}

\maketitle

\begin{abstract}
Several new families of multi-memory classical convolutional
Bose-Chaudhuri-Hocquenghem (BCH) codes as well as families of
unit-memory quantum convolutional codes are constructed in this
paper. Our unit-memory classical and quantum convolutional codes
are optimal in the sense that they attain the classical (quantum)
generalized Singleton bound. The constructions presented in this
paper are performed algebraically and not by computational search.
\end{abstract}

\textbf{\emph{Index Terms}} -- \textbf{convolutional codes,
quantum convolutional codes, MDS codes, cyclic codes}

\section{Introduction}\label{Intro}

Several works available in the literature deal with constructions of quantum error-correcting
codes (QECC) \cite{Calderbank:1998,Steane:1999,Nielsen:2000,Cohen:1999,Bierbrauer:2000,Ashikmin:2001,Grassl:2003,SK:2005,Chen:2005,Ketkar:2006,Salah:2007,Hamada:2008,LaGuardia:2009,LaGuardia:2010}.
In contrast with this subject of research one has the theory of
quantum convolutional codes
\cite{Ollivier:2003,Ollivier:2004,Aido:2004,Grassl:2005,Grassl:2006,Grassl:2007,Aly:2007,Klapp:2007,Forney:2007,Wilde:2008,Wilde:2010,Tan:2010}.
Ollivier and Tillich \cite{Ollivier:2003,Ollivier:2004} were the first to develop the stabilizer structure
for these codes. Almeida and Palazzo Jr. construct an $[(4, 1, 3)]$ (memory $m=3$) quantum convolutional code
\cite{Aido:2004}. Grassl and R\"otteler \cite{Grassl:2005,Grassl:2006,Grassl:2007} constructed
quantum convolutional codes as well as they provide algorithms to obtain non-catastrophic encoders. Forney,
in a joint work with Guha and Grassl, constructed rate $(n-2)/n$ quantum convolutional codes.
Wilde and Brun \cite{Wilde:2008,Wilde:2010} constructed entanglement-assisted quantum convolutional coding
and Tan and Li \cite{Tan:2010} constructed quantum convolutional codes derived from LDPC codes.

Constructions of (classical) convolutional codes and their
corresponding properties as well as constructions of optimal
convolutional codes (in the sense that they attain the generalized
Singleton bound \cite{Rosenthal:1999}) have been also presented in
the literature
\cite{Forney:1970,Lee:1976,Piret:1988,Rosenthal:1999,York:1999,Hole:2000,Rosenthal:2001,Luerssen:2006,Luerssen:2008}.
In particular, in the paper by Rosenthal and York
\cite{York:1999}, the authors obtained some of the matrices of the
state-space realization of the convolutional codes in the same way
as the parity check matrix of a BCH block code, generating
convolutional codes with different structures of (classical block)
BCH codes. As it is well known, the generalized (classical)
Singleton bound \cite{Rosenthal:1999} (see also
\cite{Rosenthal:2001}) appears recently in the literature. In the
paper by Piret \cite{Piret:1988art} and even in the handbook
\cite{Piret:1988}, the concept of MDS convolutional codes was
addressed, but in a different context that the previously
mentioned. In this paper we use the notion of MDS convolutional
codes according to Smarandache and Rosenthal
\cite{Rosenthal:2001}.

Keeping these facts in mind, in this paper we propose
constructions of new families of quantum and classical
convolutional codes by applying the famous method proposed by
Piret \cite{Piret:1988} and recently generalized by Aly \emph{et
al.} \cite{Aly:2007}, which consists in the construction of
(classical) convolutional codes derived from block codes. More
precisely, we first construct new families of classical
maximum-distance-separable (MDS) convolutional codes ( in the
sense that they attain the generalized Singleton bound
\cite[Theorem 2.2]{Rosenthal:1999}) as well as new families of
multi-memory convolutional codes. After these constructions, we
apply the well known technique by Aly \emph{et al.}
\cite[Proposition 2]{Aly:2007} in order to construct new MDS
convolutional stabilizer codes (in the sense that they attain the
quantum generalized Singleton bound \cite[Theorem 7]{Klapp:2007})
derived from their classical counterparts.

An advantage of our techniques of construction lie in the fact
that all new (classical and quantum) convolutional codes are
generated algebraically and not by computational search.
Therefore, new families of classical and quantum optimal
convolutional codes are constructed, not only specific codes, in
contrast with many works where only exhaustively computational
search or even specific codes are constructed.

The constructions proposed here deal with suitable properties of
cyclotomic cosets, that will be specified throughout this paper.
These nice properties of the cosets hold when considering
classical convolutional codes of length $n=q+1$ over the field
$F_{q}$ for all prime power $q$, or even quantum convolutional
codes of length $n=q^2 +1$ over $F_{q^{2}}$, where $q=2^{t}$,
$t\geq 3$ is an integer. In the quantum case, the corresponding
classical codes are endowed with the Hermitian inner product.

The new families of classical convolutional MDS codes constructed
have parameters

\begin{itemize}
\item $(n, n-2i, 2; 1, 2i+3)_{q},$ where $1\leq i \leq \frac{q}{2}
- 1$, $q=2^t$, $t\geq 3$ is an integer, $n=q+1$ is the code
length, $k=n-2i$ is the code dimension, $\gamma=2$ is the degree
of the code, $m=1$ is the memory and $d_{f}=2i+3$ is the free
distance of the code;

\item $(n, n-2i+1, 2; 1, 2i+2)_{q},$ where $q=p^{t}$, $t\geq 2$ is
an integer, $p$ is an odd prime number, $n=q+1$ and $2\leq i \leq
\frac{n}{2} - 1$.
\end{itemize}

The multi-memory (classical) convolutional codes constructed here
have parameters

\begin{itemize}
\item $(n, 2r+1, 2m; m, d_{f}\geq n-2[r+m])_{q}$, where $q=p^{t}$,
$t\geq 2$ is an integer, $p$ is an odd prime number, $n=q+1$, $r,
m$ are integers with $r \geq 1$, $m \geq 2$ and $3 \leq r + m \leq
\frac{n}{2} - 1$.
\end{itemize}

The new convolutional stabilizer MDS codes have parameters

\begin{itemize}
\item $[(n, n-4i, 1; 2 , 2i+3)]_{q},$ where $2\leq i \leq
\frac{q}{2} - 2$, $q=2^t$, $t\geq 3$ is an integer and $n=q^2+1$.
Here, $n$ is the frame size, $k=n-4i$ is the number of logical
qudits per frame, $m=1$ is the memory, $\gamma = 2$ is the degree
and $d_{f}=2i+3$ is the free distance of the code.
\end{itemize}

Note that the order between the degree and the memory are changed
when comparing the parameters of classical and quantum
convolutional codes. This notation is adopted to keep the same
notation utilized in \cite{Aly:2007}.

Let us now give the structure of the paper. In Section~\ref{II},
we review basic concepts on cyclic codes. In Section~\ref{III}, a
review of concepts concerning classical and quantum convolutional
codes is given. In Section~\ref{IV}, we propose constructions of
new families of classical MDS convolutional codes as well as
families of multi-memory convolutional codes. In Section~\ref{V}
we construct new optimal (MDS) quantum convolutional codes and, in
Section~\ref{VI}, a brief summary of this work is described.

\section{Review of Cyclic codes}\label{II}

\emph{Notation.} Throughout this paper, $p$ denotes a prime
number, $q$ is a prime power and $F_{q}$ is a finite field with
$q$ elements. The code length is denoted by $n$ and we always
assume that $\gcd (q, n) = 1$. As usual, the multiplicative order
of $q$ modulo $n$ is given by $l= \ {{ord}_{n}}(q)$, $\alpha$
denotes a primitive $n$-th root of unity, and the minimal
polynomial (over $F_{q}$) of an element ${\alpha}^{j} \in
F_{q^{m}}$ is denoted by ${M}^{(j)}(x)$.

The notation ${\mathcal{C}}_{s}$ is utilized to denote a cyclotomic coset
containing $s$, the code ${C}^{{\perp}}$ denotes the Euclidean dual
and the code ${C}^{{\perp}_{h}}$ denotes the Hermitian dual of a given code $C$.

Let $C$ be a cyclic code of length $n$ over $F_{q}$. Then there exists only one
monic polynomial $g(x)$ with minimal degree in $C$. Moreover,
$C = \langle g(x)\rangle$, i. e., $g(x)$ is a generator polynomial of $C$ and
$g(x)$ is a factor of $x^{n}-1$. The dimension of $C$ equals $n - r$, where
$r = \deg g(x)$.

\begin{theorem}(The BCH bound)\cite[pg. 201]{Macwilliams:1977}
Let $\alpha$ be a primitive $n$-th root of unity. Let $C$ be a
cyclic code with generator polynomial $g(x)$ such that, for some
integers $b\geq 0$ and $\delta\geq 1,$ and for $\alpha\in F_{q}$,
we have $g({\alpha}^{b}) = g({\alpha}^{b+1}) = \ldots =
g({\alpha}^{b+\delta-2})=0$, that is, the code has a sequence of
$\delta-1$ consecutive powers of $\alpha$ as zeros. Then the
minimum distance of $C$ is, at least, $\delta$.
\end{theorem}

\begin{definition}\cite[pg. 202]{Macwilliams:1977}
Let $q$ be a prime power and $\alpha$ be a primitive $n$-th root
of unity. A cyclic code $C$ of length $n$ over $F_{q}$ ($\gcd (q,
n) = 1$) is a BCH code with designed distance $\delta$ if, for
some integer $b\geq 0,$ we have
$$g(x)= l.c.m. \{{M}^{(b)}(x), {M}^{(b+1)}(x), \ldots,
{M}^{(b+\delta-2)}(x)\},$$ that is, $g(x)$ is the monic polynomial
of smallest degree over $F_{q}$ having ${{\alpha}^{b}},
{{\alpha}^{b+1}},\ldots, {{\alpha}^{b+\delta-2}}$ as zeros.
Therefore, $c\in C$ if and only if
$c({\alpha}^{b})=c({{\alpha}^{b+1}})=\ldots =
c({{\alpha}^{b+\delta-2}})=0$. Thus the code has a string of
$\delta - 1$ consecutive powers of $\alpha$ as zeros. A parity
check matrix for $C$ is given by

\begin{eqnarray*}
H_{\delta , b} = \\=  \left[
\begin{array}{ccccc}
1 & {{\alpha}^{b}} & {{\alpha}^{2b}} & \cdots & {{\alpha}^{(n-1)b}} \\
1 & {{\alpha}^{(b+1)}} & {{\alpha}^{2(b+1)}} & \cdots & {{\alpha}^{(n-1)(b+1)}}\\
\vdots & \vdots & \vdots & \vdots & \vdots\\
1 & {{\alpha}^{(b+\delta-2)}} & \cdots & \cdots & {{\alpha}^{(n-1)(b+\delta-2)}}\\
\end{array}
\right],
\end{eqnarray*}
where each entry is replaced by the corresponding column of $l$
elements from $F_{q}$, where $l= \ {{ord}_{n}}(q)$, and then
removing any linearly dependent rows. The rows of the resulting
matrix over $F_{q}$ are the parity checks satisfied by $C$.
\end{definition}

Let ${\mathcal B} =\{ b_{1}, \ldots, b_{l}\}$ be a basis of
$F_{q^{l}}$ over $F_{q}$. If $u = (u_1,\ldots ,u_{n}) \in
F_{q^{l}}^{n}$ then one can write the vectors $u_{i}$, $1\leq
i\leq n$, as linear combinations of the elements of ${\mathcal
B}$, that is, $u_{i} = u_{i1}b_{1} +\ldots + u_{il}b_{l}$.
Consider that $u^{(j)} = (u_{1j},\ldots, u_{nj})$ are vectors in
$F_{q}^{n}$ with $1\leq j\leq l$. Then, if $v \in F_{q}^{n}$, one
has $v\cdot u=0$ if and only if $v \cdot u^{(j)} = 0$ for all
$1\leq j\leq l$.

From the BCH bound, the minimum distance of a BCH code is greater
than or equal to its designed distance $\delta$. If $n=q^{l} - 1$
then the BCH code is called primitive and if $b=1$ it is called
narrow-sense.

\section{Review of Convolutional Codes}\label{III}

In this section we present a brief review of classical and quantum
convolutional codes. For more details we refer the reader to
\cite{Forney:1970,Piret:1988,Johannesson:1999,Huffman:2003,Aly:2007,Klapp:2007}.
The following results can be found in
\cite{Johannesson:1999,Huffman:2003,Klapp:2007,Aly:2007}.

Recall that a polynomial encoder matrix $G(D) \in F_{q}{[D]}^{k
\times n}$ is called \emph{basic} if $G(D)$ has a polynomial right
inverse. A basic generator matrix is called \emph{reduced} (or
minimal \cite{Rosenthal:2001,Huffman:2003,Luerssen:2008}) if the
overall constraint length $\gamma =\displaystyle\sum_{i=1}^{k}
{\gamma}_i$ has the smallest value among all basic generator
matrices (in this case the overall constraint length $\gamma$ will
be called the \emph{degree} of the resulting code).

\begin{definition}\cite{Klapp:2007}
A rate $k/n$ \emph{convolutional code} $C$ with parameters $(n, k,
\gamma ; m, d_{f} {)}_{q}$ is a submodule of $F_q {[D]}^{n}$
generated by a reduced basic matrix $G(D)=(g_{ij}) \in F_q
{[D]}^{k \times n}$, that is, $C = \{ {\bf u}(D)G(D) | {\bf
u}(D)\in F_{q} {[D]}^{k} \}$, where $n$ is the length, $k$ is the
dimension, $\gamma =\displaystyle\sum_{i=1}^{k} {\gamma}_i$ is the
\emph{degree}, where ${\gamma}_i = {\max}_{1\leq j \leq n} \{ \deg
g_{ij} \}$, $m = {\max}_{1\leq i\leq k}\{{\gamma}_i\}$ is the
\emph{memory} and $d_{f}=$wt$(C)=\min \{wt({\bf v}(D)) \mid {\bf
v}(D) \in C, {\bf v}(D)\neq 0 \}$ is the \emph{free distance} of
the code.
\end{definition}

In the above definition, the \emph{weight} of an element ${\bf
v}(D)\in F_{q} {[D]}^{n}$ is defined as wt$({\bf
v}(D))=\displaystyle\sum_{i=1}^{n}$wt$(v_i(D))$, where
wt$(v_i(D))$ is the number of nonzero coefficients of $v_{i}(D)$.

If one considers the field of Laurent series $F_{q}((D))$ whose
elements are given by ${\bf u}(D) = {\sum}_{i} u_i D^{i}$, where
$u_i \in F_{q}$ and $u_i = 0$ for $i\leq r $, for some $r \in
\mathbb{Z}$, we define the weight of ${\bf u}(D)$ as wt$({\bf
u}(D)) = {\sum}_{\mathbb{Z}}$wt$(u_i)$. A generator matrix $G(D)$
is called \emph{catastrophic} if there exists a ${\bf
u}{(D)}^{k}\in F_{q}{((D))}^{k}$ of infinite Hamming weight such
that ${\bf u}{(D)}^{k}G(D)$ has finite Hamming weight. Since a
basic generator matrix is non-catastrophic, all the classical
(quantum) convolutional codes constructed in this paper have non
catastrophic generator matrices.

Let us recall that the Euclidean inner product of two $n$-tuples ${\bf u}(D)
= {\sum}_i {\bf u}_i D^i$ and ${\bf v}(D) = {\sum}_j {\bf u}_j D^j$ in
$F_q {[D]}^{n}$ is defined as $\langle {\bf u}(D)\mid {\bf v}(D)\rangle = {\sum}_i
{\bf u}_i \cdot {\bf v}_i$. If $C$ is a convolutional code then the code
$C^{\perp }=\{ {\bf u}(D) \in F_q {[D]}^{n}\mid \langle {\bf u}(D)\mid
{\bf v}(D)\rangle = 0$ for all ${\bf v}(D)\in C\}$ denotes its Euclidean dual.

Similarly, the Hermitian inner product is defined as
$\langle {\bf u}(D)\mid {\bf v}(D){\rangle}_{h} = {\sum}_i
{\bf u}_i \cdot {\bf v}_i^{q}$, where ${\bf u}_i, {\bf v}_i\in F_{q^{2}}^{n}$
and ${\bf v}_i^{q} = (v_{1i}^{q}, \ldots, v_{ni}^{q})$. The Hermitian dual
of the code $C$ is defined by $C^{{\perp}_{h} }=\{ {\bf u}(D) \in
F_{q^{2}} {[D]}^{n}\mid \langle {\bf u}(D)\mid
{\bf v}(D){\rangle}_{h} = 0$ for all ${\bf v}(D)\in C\}$.

\subsection{Convolutional Codes Derived from Block
Codes}\label{IIIA}

In this subsection we recall some results shown in \cite{Aly:2007} that will be
utilized in the proposed constructions.

We consider that ${[n, k, d]}_{q}$ is a block code with parity check matrix $H$ and then
we split $H$ into $m+1$ disjoint submatrices $H_i$ such that
\begin{eqnarray}
H = \left[
\begin{array}{c}
H_0\\
H_1\\
\vdots\\
H_{m}\\
\end{array}
\right],
\end{eqnarray}
where each $H_i$ has $n$ columns, obtaining the polynomial matrix
\begin{eqnarray}
G(D) =  {\tilde H}_0 + {\tilde H}_1 D + {\tilde H}_2 D^2 + \ldots
+ {\tilde H}_m D^m,
\end{eqnarray}
where the matrices ${\tilde H}_i$, for all $1\leq i\leq m$, are
derived from the respective matrices $H_i$ by adding zero-rows at
the bottom in such a way that the matrix ${\tilde H}_i$ has
$\kappa$ rows in total, where $\kappa$ is the maximal number of
rows among the matrices $H_i$. As it is well known, the matrix
$G(D)$ generates a convolutional code with $\kappa$ rows. Note
that $m$ is the memory of the resulting convolutional code
generated by the matrix $G(D)$.

\begin{theorem}\cite[Theorem 3]{Aly:2007}\label{A}
Suppose that $C \subseteq F_q^n$ is a linear code with parameters
${[n, k, d]}_{q}$ and assume also that $H \in F_q^{(n-k)\times n}$
is a parity check matrix for $C$ partitioned into submatrices
$H_0, H_1, \ldots, H_m$ as in eq.~(1) such that $\kappa =$ rk$H_0$
and rk$H_i \leq \kappa$
for $1 \leq i\leq m$ and consider the polynomial matrix $G(D)$ as in eq.~(2). Then we have:\\
(a) The matrix $G(D)$ is a reduced basic generator matrix;\\
(b) If $C^{\perp}\subset C$ (resp. ${C}^{{\perp}_{h}}\subset C$),
then the convolutional code $V = \{ {\bf v}(D) = {\bf u}(D)G(D)
\mid {\bf u}(D)\in F_q^{n-k} [D] \}$ satisfies $V \subset V^{\perp}$ (resp. $V \subset {V}^{{\perp}_{h}}$);\\
(c) If $d_f$ and $d_f^{\perp}$ denote the free distances of $V$
and $V^{\perp}$, respectively, $d_i$ denote the minimum distance
of the code $C_i = \{ {\bf v}\in F_q^n \mid {\bf v} {\tilde H}_i^t
=0 \}$ and $d^{\perp}$ is the minimum distance of $C^{\perp}$,
then one has $\min \{ d_0 + d_m , d \} \leq d_f^{\perp} \leq  d$
and $d_f \geq d^{\perp}$.
\end{theorem}


\subsection{Review of Quantum Convolutional Codes}\label{IIIB}

We begin this subsection by describing briefly the concept of
quantum convolutional codes. For more details the reader can
consult \cite{Ollivier:2004}.

A quantum convolutional code is defined by means of its stabilizer
which is a subgroup of the infinite version of the Pauli group,
consisting of tensor products of generalized Pauli matrices acting
on a semi-infinite stream of qudits. The stabilizer can be defined
by a stabilizer matrix of the form
$$S(D) = ( X(D)\mid Z(D)) \in F_{q}{[D]}^{(n-k)\times 2n}$$
satisfying $X(D){Z(1/D)}^{t} - Z(D){X(1/D)}^{t}=0$ (symplectic
orthogonality). More precisely, consider a quantum convolutional
code $C$ defined by a full-rank stabilizer matrix $S(D)$ given
above. Then $C$ is a rate $k/n$ code with parameters $[(n, k, m;
\gamma, d_{f} ){]}_{q}$, where $n$ is the frame size, $k$ is the
number of logical qudits per frame, $m = {\max}_{1\leq i\leq n-k,
1\leq j\leq n} \{ \max \{ \deg {X}_{ij}(D), \deg {Z}_{ij} (D)\}
\}$ is the memory, $d_{f}$ is the free distance and $\gamma$ is
the degree of the code. Similarly as in the classical case, the
constraint lengths are defined as ${\gamma}_{i} = {\max}_{1\leq
j\leq n} \{ \max \{\deg X_{ij}(D), \deg Z_{ij}(D)\} \}$, and the
overall constraint length is defined as $\gamma
=\displaystyle\sum_{i=1}^{n-k} {\gamma}_{i}$.

On the other hand, a quantum convolutional code can also be
described in terms of a semi-infinite stabilizer matrix $S$ with
entries in $F_{q} \times F_{q}$ in the following way. If
$S(D)=\displaystyle\sum_{i=0}^{m}G_{i}D^{i}$, where each matrix
$G_{i}$ for all $i=0, \ldots, m$, is a matrix of size $(n - k)
\times n$, then the semi-infinite matrix is defined as
\begin{eqnarray*}
S = \left[
\begin{array}{cccccccc}
G_0 & G_1 & \ldots & G_{m} & 0 & \ldots & \ldots & \ldots\\
0 & G_0 & G_1 & \ldots & G_{m} & 0 & \ldots & \ldots\\
0 & 0 & G_0 & G_1 & \ldots & G_{m} & 0 & \ldots\\
\vdots & \vdots & \vdots & \vdots & \vdots & \vdots & \vdots & \vdots\\
\end{array}
\right].
\end{eqnarray*}

Next, let ${\mathbb H} = {\mathbb C}^{q^n} = {\mathbb C}^{q}
\otimes \ldots \otimes {\mathbb C}^{q}$ be the Hilbert space and
$\mid$$x \rangle$ be the vectors of an orthonormal basis of
${\mathbb C}^{q}$, where the labels $x$ are elements of $F_{q}$.
Consider $a, b \in F_{q}$ and take the unitary operators $X(a)$
and $Z(b)$ in ${\mathbb C}^{q}$ defined by $X(a)$$\mid$$x \rangle
=$$\mid$$x + a\rangle$ and $Z(b)$$\mid$$x \rangle =
w^{tr(bx)}$$\mid$$x\rangle$, respectively, where $w=\exp (2\pi i/
p)$ is a primitive $p$-th root of unity, $p$ is the characteristic
of $F_{q}$ and $tr$ is the trace map from $F_{q}$ to $F_{p}$.
Considering the \emph{error basis} ${\mathbb E} = \{X(a), Z(b) |
a, b \in F_{q} \}$, one defines the set $P_{\infty}$ (according to
\cite{Klapp:2007}) as the set of all infinite tensor products of
matrices $N\in \langle M\mid M \in {\mathbb E} \rangle$, in which
all but finitely many tensor components are equal to $I$, where
$I$ is the $q\times q$ identity matrix. Then one defines the
\emph{weight} wt of $A\in P_{\infty}$ as its (finite) number of
nonidentity tensor components. In this context, one says that a
quantum convolutional code has free distance $d_{f}$ if and only
if it can detect all errors of weight less than $d_{f}$, but
cannot detect some error of weight $d_{f}$.

The following lemma deals with the existence of convolutional
stabilizer codes derived from classical convolutional codes:

\begin{lemma}\cite[Proposition 2]{Aly:2007}\label{BB}
Let $C$ be an ${(n, (n - k)/2, \gamma; m)}_{q^2}$ convolutional code such that
$C \subseteq {C}^{{\perp}_{h}}$. Then there exists an ${[(n, k, m;
\gamma, d_f)]}_{q}$ convolutional stabilizer code, where $d_f=$ wt$({C}^{{\perp}_{h}}
\backslash C)$.
\end{lemma}

In \cite{Klapp:2007}, the authors derived the quantum
\emph{Singleton} bound for quantum convolutional codes as it is
shown in the next theorem. Let $C$ be an $[(n, k, m; \gamma, d_{f}
){]}_{q}$ quantum convolutional code. Recall that $C$ is a
\emph{pure} code if does not exist errors of weight less than
$d_{f}$ in the stabilizer of $C$.

\begin{theorem}\label{SingC}
(Quantum Singleton bound) The free distance of an ${[(n, k, m; \gamma, d_{f})]}_{q}$ $F_{q^{2}}$-linear
pure convolutional stabilizer code is bounded by
\begin{eqnarray*}
d_{f} \leq \frac{n-k}{2} \left(\left\lfloor \frac{2\gamma}{n+k}\right \rfloor + 1\right) + \gamma + 1.
\end{eqnarray*}
\end{theorem}

\begin{remark}
When Klappenecker \emph{et al.} introduced the generalized quantum
Singleton bound (GQSB) (see \cite{Klapp:2007}) they developed an
approach to convolutional stabilizer codes based on direct limit
constructions. It seems that the direct limit structure behaves
well with respect to the trace-alternant form. In this context
they derived the GQSB. It is interesting to note that this is one
of few bounds presenting in the literature concerning quantum
convolutional codes.
\end{remark}


\section{New Classical MDS-Convolutional Codes}\label{IV}

Constructions of classical convolutional codes with good or even
optimal parameters (where the latter class of codes is known as
maximum-distance-separable or MDS codes, i.e., codes attaining the
generalized Singleton bound according to \cite{Rosenthal:1999}) is
a difficult task
\cite{Piret:1988,Rosenthal:1999,Hole:2000,Rosenthal:2001,Hutchinson:2005,Gluesing:2006,Schmale:2006,Luerssen:2006,Luerssen:2008,Climent:2008,Iglesias:2009}.
Due to this difficulty, most of methods available in the
literature are based on computational search. Keeping in mind the
discussion above, our purpose is to construct new families of
classical and quantum MDS convolutional codes by applying
algebraic methods.

The main results of this section are Theorem~\ref{main} and
Theorem~\ref{mainI}. They generate new families of optimal (in the
sense that the codes attain the generalized Singleton bound
\cite{Rosenthal:1999}) convolutional codes of length $n=q + 1$,
over $F_{q}$ for all prime power $q$. Before proceeding further,
recall the well known result from \cite{Macwilliams:1977}:

\begin{lemma}\cite[Theorem 9, Chapter 11]{Macwilliams:1977})\label{nicecyclo}
Suppose that $q=2^t$, where $t\geq 2$ is an integer, $n=q+1$ and
consider that $a=\frac{q}{2}$. Then one has:
\begin{enumerate}
\item[ i)] With exception of coset ${\mathcal{C}}_{0}=\{0\}$, each one of the
other $q$-ary cyclotomic cosets is of the form ${\mathcal{C}}_{a-i}=
\{ a-i, a+i+1 \}$, where $0\leq i\leq a-1$;

\item[ ii)] The $q$-ary cosets ${\mathcal{C}}_{a-i}=\{ a-i, a+i+1
\}$, where $0\leq i\leq a-1$, are mutually disjoint.

\end{enumerate}
\end{lemma}

We are now able to show one of the main results of this section:

\begin{theorem}\label{main}
Assume that $q=2^{t}$, where $t\geq 3$ is an integer, $n=q+1$ and
consider that $a=\frac{q}{2}$. Then there exist classical MDS
convolutional codes with parameters $(n, n-2i, 2; 1, 2i+3)_{q}$,
where $1\leq i \leq a - 1$.
\end{theorem}

\begin{proof}
We first note that $ \gcd (n, q) = 1$ and ${{ord}_{n}}(q)=2$. The
proof consists of two steps. The first one is the construction of
suitable BCH (block) codes and the second step is the construction
of convolutional BCH codes derived from the BCH (block) codes
generated in the first step.

Let us begin the first step. Let $C_2$ be the BCH code of length
$n$ over $F_{q}$ generated by the product of the minimal
polynomials
\begin{eqnarray*}
C_2 = \langle g_{2} (x)  \rangle= \\ =\langle {M}^{(a-i)}(x)
{M}^{(a-i+1)}(x) \cdot \ldots \cdot {M}^{(a-1)}(x){M}^{(a)}(x)
\rangle.
\end{eqnarray*}

A parity check matrix of $C_2$ is obtained from the matrix

\begin{eqnarray*}
H_{2i+3 , a-i} = \\ = \left[
\begin{array}{ccccc}
1 & {{\alpha}^{(a-i)}} & {{ \alpha}^{2(a-i)}} & \cdots & {{\alpha}^{(n-1)(a-i)}} \\
1 & {{\alpha}^{(a-i+1)}} & {{\alpha}^{2(a-i+1)}} & \cdots & {{\alpha}^{(n-1)(a-i+1)}}\\
\vdots & \vdots & \vdots & \vdots & \vdots\\
1 & {{\alpha}^{(a-1)}} & \cdots & \cdots & {{\alpha}^{(n-1)(a-1)}}\\
1 & {{\alpha}^{a}} & \cdots & \cdots & {{\alpha}^{(n-1)a}}\\
\end{array}
\right]
\end{eqnarray*}
by expanding each entry as a column vector (containing $2$ rows)
with respect to some $F_{q}-$basis $\beta$ of $F_{q^2}$ and then
removing any linearly dependent rows. This new matrix $H_{C_2}$ is
a parity check matrix of $C_{2}$ and it has $2i+2$ rows. Since the
dimension of $C_{2}$ is equal to $n - 2(i+1)$ (as proved in the
paragraph below), so there is no linearly dependent rows in
$H_{C_2}$.

From Lemma~\ref{nicecyclo}, each one of the $q$-ary cyclotomic
cosets ${\mathcal{C}}_{a-i}$, where $0\leq i\leq a-1$
(corresponding to the minimal polynomials ${M}^{(a-i)}(x)$), has
two elements and they are mutually disjoint. Since the degree of
the generator polynomial $g_{2} (x)$ of the code $C_{2}$ equals
the cardinality of its defining set, then one has $\deg (g_{2}
(x))=2(i+1)$, so the dimension $k_{C_{2}}$ of $C_{2}$ equals
$k_{C_{2}}= n -\deg (g_{2} (x))= n - 2(i+1)$. Moreover, the
defining set of the code $C_{2}$ consists of the sequence $\{ a-i,
a-i+1, \ldots, a, a+1, \ldots, a+i+1\}$ of $2i+2$ consecutive
integers, so, from the BCH bound, the minimum distance $d_{C_{2}}$
of $C_{2}$ satisfies $d_{C_{2}} \geq 2i+3$. Thus, $C_2$ is a MDS
code with parameters ${[n, n-2i-2, 2i+3]}_{q}$ and, consequently,
its (Euclidean) dual code has dimension $2i+2$.

We next consider that $C_1$ is the BCH code of length $n$ over
$F_{q}$ generated by the product of the minimal polynomials
\begin{eqnarray*}
C_1 =\langle g_1 (x)\rangle =\\= \langle {M}^{(a-i+1)}(x)
{M}^{(a-i+2)}(x)\cdot \ldots \cdot
{M}^{(a-1)}(x){M}^{(a)}(x)\rangle.
\end{eqnarray*}

Similarly, $C_1$ has a parity check matrix derived from the matrix

\begin{eqnarray*}
H_{2i+1 , a-i+1} = \\ = \left[
\begin{array}{ccccc}
1 & {{\alpha}^{(a-i+1)}} & {{\alpha}^{2(a-i+1)}} & \cdots & {{\alpha}^{(n-1)(a-i+1)}}\\
1 & {{\alpha}^{(a-i+2)}} & {{\alpha}^{2(a-i+2)}} & \cdots & {{\alpha}^{(n-1)(a-i+2)}}\\
\vdots & \vdots & \vdots & \vdots & \vdots\\
1 & {{\alpha}^{(a-1)}} & \cdots & \cdots & {{\alpha}^{(n-1)(a-1)}}\\
1 & {{\alpha}^{a}} & \cdots & \cdots & {{\alpha}^{(n-1)a}}\\
\end{array}
\right]
\end{eqnarray*}
by expanding each entry as a column vector (containing $2$ rows)
with respect to $\beta$ (already done, since $H_{2i+1 , a-i+1}$ is
a submatrix of $H_{2i+3 , a-i}$). After performing the expansion
to all entries, such new matrix is denoted by $H_{C_1}$ ($H_{C_1}$
is a submatrix of $H_{C_2}$). Applying again Lemma~\ref{nicecyclo}
and proceeding similarly as above, it follows that $C_1$ is a MDS
code with parameters ${[n, n-2i, 2i+1]}_{q}$.

To finish the first step, consider $C$ be the BCH code of length
$n$ over $F_{q}$ generated by the minimal polynomial
${M}^{(a-i)}(x)$, that is,
\begin{eqnarray*}
C = \langle {M}^{(a-i)}(x)\rangle.
\end{eqnarray*}
$C$ has parameters ${[n, n-2, d\geq 2]}_{q}$. A parity check
matrix $H_{C}$ of $C$ is given by expanding each entry of the
matrix
\begin{eqnarray*}
H_{2, a-i} = \\ = \left[
\begin{array}{ccccc}
1 & {{\alpha}^{(a-i)}} & {{\alpha}^{2(a-i)}} & \cdots & {{\alpha}^{(n-1)(a-i)}} \\
\end{array}
\right]
\end{eqnarray*}
with respect to $\beta$ (already done, since $H_{2, a-i}$ is a
submatrix of $H_{2i+3 , a-i}$). Since $C$ has dimension $n-2$,
$H_{C}$ has rank $2$ ($H_{C}$ is also a submatrix of $H_{C_2}$).

Next we describe the second step. We begin by rearranging the rows
of $H_{C_2}$ in the form
\begin{eqnarray*}
H = \\ = \left[
\begin{array}{ccccc}
1 & {{\alpha}^{a}} & \cdots & \cdots & {{\alpha}^{(n-1)a}}\\
1 & {{\alpha}^{(a-1)}} & \cdots & \cdots & {{\alpha}^{(n-1)(a-1)}}\\
\vdots & \vdots & \vdots & \vdots & \vdots\\
1 & {{\alpha}^{(a-i+1)}} & {{\alpha}^{2(a-i+1)}} & \cdots & {{\alpha}^{(n-1)(a-i+1)}}\\
1 & {{\alpha}^{(a-i)}} & {{\alpha}^{2(a-i)}} & \cdots & {{ \alpha}^{(n-1)(a-i)}} \\
\end{array}
\right],
\end{eqnarray*}
(to simplify the notation we write $H$ in terms of powers of
$\alpha$, although it is clear from the context that this matrix
has entries in $F_{q}$, which are derived from expanding each
entry with respect to the basis $\beta$ already performed).

Then we split $H$ into two disjoint submatrices $H_0$ and $H_1$ of
the forms

\begin{eqnarray*}
H_0 = \\ = \left[
\begin{array}{ccccc}
1 & {{\alpha}^{a}} & \cdots & \cdots & {{\alpha}^{(n-1)a}}\\
1 & {{\alpha}^{(a-1)}} & \cdots & \cdots & {{\alpha}^{(n-1)(a-1)}}\\
\vdots & \vdots & \vdots & \vdots & \vdots\\
1 & {{\alpha}^{(a-i+1)}} & {{\alpha}^{2(a-i+1)}} & \cdots & {{\alpha}^{(n-1)(a-i+1)}}\\
\end{array}
\right]
\end{eqnarray*}
and
\begin{eqnarray*}
H_1 = \\ = \left[
\begin{array}{ccccc}
1 & {{\alpha}^{(a-i)}} & {{\alpha}^{2(a-i)}} & \cdots & {{\alpha}^{(n-1)(a-i)}} \\
\end{array}
\right],
\end{eqnarray*}
respectively, where $H_0$ is obtained from the matrix $H_{C_1}$ by
rearranging rows and $H_1$ is derived from $H_{C}$ also by
rearranging rows. Hence it follows that rk$H_0\geq$ rk$H_1$.

Then we form the convolutional code $V$ generated by the reduced
basic (according to Theorem~\ref{A} Item (a)) generator matrix
\begin{eqnarray*}
G(D)=\tilde H_{0}+ \tilde H_1 D,
\end{eqnarray*}
where $\tilde H_{0} = H_0$ and $\tilde H_1$ is obtained from $H_1$
by adding zero-rows at the bottom such that $\tilde H_1$ has the
number of rows of $H_0$ in total. By construction, $V$ is a
unit-memory convolutional code of dimension $2i$ and degree
${\delta}_{V} = 2$.

Consider next the Euclidean dual $V^{\perp}$ of the convolutional
code $V$. We know that $V^{\perp}$ has dimension $n-2i$ and degree
$2$. Let us now compute the free distance $d_{f}^{\perp}$ of
$V^{\perp}$. By Theorem~\ref{A} Item (c), the free distance of
$V^{\perp}$ is bounded by $\min \{ d_0 + d_1 , d \} \leq
d_{f}^{\perp} \leq  d$, where $d_i$ is the minimum distance of the
code $C_i = \{ {\bf v}\in F_q^n \mid {\bf v} {\tilde H}_i^t =0
\}$. From construction one has $d = 2i+3$, $d_0 = 2i+1$ and $d_1
\geq 2$, so $V^{\perp}$ has parameters $(n, n-2i, 2; 1,
2i+3)_{q}$.

Recall that the generalized (classical) Singleton bound
\cite{Rosenthal:2001} of an $(n, k, \gamma ; m, d_{f} {)}_{q}$
convolutional code is given by
$$ d_{f}\leq (n-k)[ \lfloor \gamma/k \rfloor +1 ] + \gamma +1.$$ Replacing the values
of the parameters of $V^{\perp}$ in the above inequality one
concludes that $V^{\perp}$ is a MDS convolutional code and the
proof is complete.
\end{proof}

\begin{remark}
Note that the new codes have degree $\gamma=2$. The reason for
this is as follows: in order to obtain codes with maximum minimum
distances we have to construct codes (the notation is the same
utilized in Theorem~\ref{main}) satisfying the inequalities $\min
\{ d_0 + d_1 , d \} \leq d_f^{\perp} \leq d$. Therefore one
designs the code $C$ with parameters ${[n, n-2, d_1\geq 2]}_{q}$.
Now, it is easy to see that the corresponding convolutional code
$V^{\perp}$ has degree $2$.
\end{remark}

Let us now give an illustrative example.

\begin{example}
According to Theorem~\ref{main}, let $q=16$, $n=q+1=17$ and $a=8$.
Assume $C_2$ is an ${[17, 11, 7]}_{16}$ (cyclic) MDS code
generated by the product of the minimal polynomials
$M^{(8)}(x)M^{(7)}(x)M^{(6)}(x)$. The corresponding cyclotomic
cosets of $C_2$ are $\{8, 9\}$, $\{7, 10 \}$ and $\{6, 11 \}$.
Consider $C_1$ be the (cyclic) MDS code generated by the product
of the minimal polynomials $M^{(8)}(x)M^{(7)}(x)$; $C_1$ has
parameters ${[17, 13, 5]}_{16}$. Finally, suppose $C$ is the
cyclic code generated by $M^{(6)}(x)$, where $C$ has parameters
${[17, 15, d\geq 2]}_{16}$. In this case we have $i=2$. Then we
can form the convolutional code $V$ with reduced basic generator
matrix $G(D)=\tilde H_{0}+ \tilde H_1 D,$ where $\tilde H_{0} =
H_0$ and $\tilde H_1$ is obtained from $H_1$ by adding zero-rows
at the bottom such that $\tilde H_1$ has the number of rows of
$H_0$ in total. The matrix $H_0$ is the parity check matrix of
$C_1$ (up to permutation of rows) and $H_1$ is the parity check
matrix of $C$. $V$ has parameters ${(17, 4, 2; 1, d_f)}_{16}$. The
Euclidean dual $V^{\perp}$ has parameters ${(17, 13, 2; 1,
d_{f}^{\perp})}_{16}$, where $\min \{ d_0 + d_1 , d \} \leq
d_{f}^{\perp} \leq  d$, where $d_0 =5$, $d_1 \geq 2$ and $d = 7$.
Therefore $V^{\perp}$ has parameters ${(17, 13, 2; 1, 7)}_{16}$.
Applying the generalized Singleton bound one has $7 = 4(\lfloor
2/13 \rfloor +1) + 2 + 1,$ so $V^{\perp}$ is MDS.
\end{example}

It is well known (see for example \cite{Rosenthal:1999}) that if a convolutional
code $C$ is MDS then one can not guarantee that its dual also is MDS.
Unfortunately in the above construction, although the codes $V^{\perp}$ are MDS, there
is no guarantee that their duals $V$ are MDS:

\begin{corollary}\label{C}
Assume $q=2^t$, where $t\geq 3$ is an integer, $n=q+1$ and
consider that $a=\frac{q}{2}$. Then there exist classical
convolutional codes with parameters $(n, 2i, 2; 1, d_{f})_{q}$,
where $1\leq i \leq a - 1$ and $d_{f}\geq n-2i-1$.
\end{corollary}
\begin{proof}
Consider the same construction and notation used in
Theorem~\ref{main}. We know that $V$ has parameters $(n, 2i, 2; 1,
d_f)_{q}$. Let us compute $d_f$. From Theorem~\ref{A} Item (b),
$d_f \geq d^{\perp}$. We know that the matrix $H$ is obtaining by
rearranging the rows of $H_{C_2}$ and the code $C_2^{\perp}$ is a
MDS code with parameters ${[n, 2i+2, n-2i-1]}_{q}$. Thus $d_{f}
\geq n-2i-1$ and $V$ has parameters $(n, 2i, 2; 1, d_{f} )_{q}$,
where $d_{f} \geq n-2i-1$.
\end{proof}

Theorem~\ref{mainI}, given in the sequence, is the second main
result of this section. More precisely, in such theorem, we
construct new families of (classical) MDS convolutional codes over
$F_{q}$ for all $q=p^{t}$, where $t\geq 2$ and $p$ is an odd prime
number. In order to prove it, we need the following well known
result:

\begin{lemma}\cite[Theorem 9, Chapter 11]{Macwilliams:1977})\label{nicecycloqary}
Suppose that $q=p^{t}$, where $t\geq 2$ is an integer and $p$ is
an odd prime number. Let $n=q+1$ and consider that
$a=\frac{n}{2}$. Then one has:
\begin{enumerate}

\item[ i)] The $q$-ary coset ${\mathcal{C}}_{a}$ has only one
element, that is, ${\mathcal{C}}_{a}=\{a\}$;

\item[ ii)] With exception of cosets ${\mathcal{C}}_{0}=\{0\}$ and
${\mathcal{C}}_{a}$, each one of the other $q$-ary cyclotomic
cosets is of the form ${\mathcal{C}}_{a-i}= \{ a-i, a+i\}$, where
$1\leq i\leq a-1$;

 \item[ iii)] The $q$-ary cosets ${\mathcal{C}}_{a-i}=\{ a-i, a+i \}$, where
$1\leq i\leq a-1$, are mutually disjoint and have two elements.

\end{enumerate}
\end{lemma}

Let us now prove Theorem~\ref{mainI}. Since its proof is analogous
to that of Theorem~\ref{main}, we only give a sketch of it.

\begin{theorem}\label{mainI}
Assume that $q=p^{t}$, where $t\geq 2$ is an integer and $p$ is an
odd prime number. Consider that $n=q+1$ and $a=\frac{n}{2}$. Then
there exist classical MDS convolutional codes with parameters $(n,
n-2i+1, 2; 1, 2i+2)_{q}$, where $2\leq i \leq a - 1$.
\end{theorem}

\begin{proof}
Let $C_2$ be the BCH code of length $n$ over $F_{q}$ generated by
the product of the minimal polynomials
\begin{eqnarray*}
C_2 = \langle g_{2} (x)  \rangle =\langle {M}^{(a-i)}(x)
{M}^{(a-i+1)}(x)\cdot \\ \ldots \cdot {M}^{(a-1)}(x){M}^{(a)}(x)
\rangle.
\end{eqnarray*}
whose parity check matrix $H_{C_2}$ is obtained from the matrix

\begin{eqnarray*}
H_{2i+2 , a-i} = \\ = \left[
\begin{array}{ccccc}
1 & {{\alpha}^{(a-i)}} & {{ \alpha}^{2(a-i)}} & \cdots & {{\alpha}^{(n-1)(a-i)}} \\
1 & {{\alpha}^{(a-i+1)}} & {{\alpha}^{2(a-i+1)}} & \cdots & {{\alpha}^{(n-1)(a-i+1)}}\\
\vdots & \vdots & \vdots & \vdots & \vdots\\
1 & {{\alpha}^{(a-1)}} & \cdots & \cdots & {{\alpha}^{(n-1)(a-1)}}\\
1 & {{\alpha}^{a}} & \cdots & \cdots & {{\alpha}^{(n-1)a}}\\
\end{array}
\right]
\end{eqnarray*}
by expanding each entry as a column vector over some $F_{q}-$basis
$\beta$ of $F_{q^2}$ and removing one linearly dependent row,
because $H_{C_2}$ has rank $2i+1$ (computed below).

From Lemma~\ref{nicecycloqary}, each one of the $q$-ary cyclotomic
cosets ${\mathcal{C}}_{a-i}$, where $2\leq i\leq a-1$, has two
elements, they are mutually disjoint and the coset
${\mathcal{C}}_{a}$ has only one element. Thus the dimension
$k_{C_{2}}$ of $C_{2}$ equals $k_{C_{2}}= n -\deg (g_{2} (x))= n -
2i-1$. Moreover, since the defining set of the code $C_{2}$
consists of the sequence $\{ a-i, a-i+1, \ldots, a, a+1, \ldots,
a+i\}$ of $2i+1$ consecutive integers then the minimum distance
$d_{C_{2}}$ of $C_{2}$ satisfies $d_{C_{2}} \geq 2i+2$. Hence,
$C_2$ is a MDS code with parameters ${[n, n-2i-1, 2i+2]}_{q}$.

We next consider $C_1$ as the BCH code of length $n$ over $F_{q}$
generated by the product of the minimal polynomials
\begin{eqnarray*}
C_1 =\langle g_1 (x)\rangle = \langle {M}^{(a-i+1)}(x)
{M}^{(a-i+2)}(x)\cdot \\ \ldots \cdot
{M}^{(a-1)}(x){M}^{(a)}(x)\rangle.
\end{eqnarray*}
whose parity check matrix $H_{C_1}$ is derived from the matrix

\begin{eqnarray*}
H_{2i , a-i+1} = \\ = \left[
\begin{array}{ccccc}
1 & {{\alpha}^{(a-i+1)}} & {{\alpha}^{2(a-i+1)}} & \cdots & {{\alpha}^{(n-1)(a-i+1)}} \\
1 & {{\alpha}^{(a-i+2)}} & {{\alpha}^{2(a-i+2)}} & \cdots & {{\alpha}^{(n-1)(a-i+2)}}\\
\vdots & \vdots & \vdots & \vdots & \vdots\\
1 & {{\alpha}^{(a-1)}} & \cdots & \cdots & {{\alpha}^{(n-1)(a-1)}}\\
1 & {{\alpha}^{a}} & \cdots & \cdots & {{\alpha}^{(n-1)a}}\\
\end{array}
\right]
\end{eqnarray*}
by expanding each entry as a column vector with respect to $\beta$
of $F_{q^2}$. Then it follows that $C_1$ is a MDS code with
parameters ${[n, n-2i+1, 2i]}_{q}$ and $H_{C_1}$ has rank $2i-1$.

Assume that $C$ is the BCH code generated by the minimal
polynomial ${M}^{(a-i)}(x)$. Then $C$ has parameters ${[n, n-2,
d\geq 2]}_{q}$. A parity check matrix $H_{C}$ of $C$ is given by
expanding each entry of the matrix
\begin{eqnarray*}
H_{2, a-i} = \\ = \left[
\begin{array}{ccccc}
1 & {{\alpha}^{(a-i)}} & {{\alpha}^{2(a-i)}} & \cdots & {{\alpha}^{(n-1)(a-i)}}\\
\end{array}
\right]
\end{eqnarray*}
with respect to $\beta$. $H_{C}$ has rank $2$.

Rearranging the rows of $H_{C_2}$ we obtain the matrix
\begin{eqnarray*}
H = \\ = \left[
\begin{array}{ccccc}
1 & {{\alpha}^{a}} & \cdots & \cdots & {{\alpha}^{(n-1)a}}\\
1 & {{\alpha}^{(a-1)}} & \cdots & \cdots & {{\alpha}^{(n-1)(a-1)}}\\
\vdots & \vdots & \vdots & \vdots & \vdots\\
1 & {{\alpha}^{(a-i+1)}} & {{\alpha}^{2(a-i+1)}} & \cdots & {{\alpha}^{(n-1)(a-i+1)}}\\
1 & {{\alpha}^{(a-i)}} & {{\alpha}^{2(a-i)}} & \cdots & {{ \alpha}^{(n-1)(a-i)}} \\
\end{array}
\right],
\end{eqnarray*}
where $a=\frac{n}{2}$. Next we split $H$ into two disjoint
submatrices $H_0$ and $H_1$ (as in Theorem~\ref{main}) of the form

\begin{eqnarray*}
H_0 = \\ = \left[
\begin{array}{ccccc}
1 & {{\alpha}^{a}} & \cdots & \cdots & {{\alpha}^{(n-1)a}}\\
1 & {{\alpha}^{(a-1)}} & \cdots & \cdots & {{\alpha}^{(n-1)(a-1)}}\\
\vdots & \vdots & \vdots & \vdots & \vdots\\
1 & {{\alpha}^{(a-i+1)}} & {{\alpha}^{2(a-i+1)}} & \cdots & {{\alpha}^{(n-1)(a-i+1)}}\\
\end{array}
\right]
\end{eqnarray*}
and
\begin{eqnarray*}
H_1 = \\ = \left[
\begin{array}{ccccc}
1 & {{\alpha}^{(a-i)}} & {{\alpha}^{2(a-i)}} & \cdots & {{\alpha}^{(n-1)(a-i)}}\\
\end{array}
\right],
\end{eqnarray*}
obtaining, in this way, the convolutional code $V$ generated by
the matrix
\begin{eqnarray*}
G(D)=\tilde H_{0}+ \tilde H_1 D
\end{eqnarray*}
with parameters $(n, 2i-1, 2; 1, d_{f})_{q}$. Proceeding similarly
as in Theorem~\ref{main}, one has a MDS convolutional code
$V^{\perp}$ with parameters $(n, n-2i+1, 2; 1, 2i+2)_{q}$, for all
$2\leq i \leq a - 1$.
\end{proof}

In the next result, we construct memory-two convolutional codes:

\begin{theorem}\label{mainII}
Assume that $q=p^{t}$, where $t\geq 2$ is an integer and $p$ is an
odd prime number. Consider that $n=q+1$ and $a=\frac{n}{2}$. Then
there exist convolutional codes with parameters $(n, 2i-3, 4; 2,
d_{f}\geq n-2i)_{q}$, where $3\leq i \leq a - 1$.
\end{theorem}

\begin{proof}
Let $C_3$ be the BCH code of length $n$ over $F_{q}$ generated by
the product of the minimal polynomials
\begin{eqnarray*}
C_3 = \langle g_3 (x)  \rangle =\langle {M}^{(a-i)}(x)
{M}^{(a-i+1)}(x) {M}^{(a-i+2)}(x)\cdot\\ \ldots\cdot
{M}^{(a-1)}(x){M}^{(a)}(x) \rangle.
\end{eqnarray*}
whose parity check matrix $H_{C_3}$ is obtained from the matrix

\begin{eqnarray*}
H_{2i+2 , a-i} = \\ = \left[
\begin{array}{ccccc}
1 & {{\alpha}^{(a-i)}} & {{ \alpha}^{2(a-i)}} & \cdots & {{\alpha}^{(n-1)(a-i)}} \\
1 & {{\alpha}^{(a-i+1)}} & {{\alpha}^{2(a-i+1)}} & \cdots & {{\alpha}^{(n-1)(a-i+1)}}\\
1 & {{\alpha}^{(a-i+2)}} & {{\alpha}^{2(a-i+2)}} & \cdots & {{\alpha}^{(n-1)(a-i+2)}}\\
\vdots & \vdots & \vdots & \vdots & \vdots\\
1 & {{\alpha}^{(a-1)}} & \cdots & \cdots & {{\alpha}^{(n-1)(a-1)}}\\
1 & {{\alpha}^{a}} & \cdots & \cdots & {{\alpha}^{(n-1)a}}\\
\end{array}
\right]
\end{eqnarray*}
by expanding each entry as a column vector over some $F_{q}-$basis
$\beta$ of $F_{q^2}$. We know that $C_3$ is a MDS code with
parameters ${[n, n-2i-1, 2i+2]}_{q}$ and $H_{C_3}$ has rank
$2i+1$.

We next consider $C_2$ as the BCH code of length $n$ over $F_{q}$
generated by the product of the minimal polynomials
\begin{eqnarray*}
C_2 =\langle g_2 (x)\rangle =\\= \langle {M}^{(a-i+2)}(x)\cdot
\ldots \cdot {M}^{(a-1)}(x){M}^{(a)}(x)\rangle.
\end{eqnarray*}
whose parity check matrix $H_{C_2}$ is derived from the matrix
\begin{eqnarray*}
H_{2i-2, a-i+2} = \\ = \left[
\begin{array}{ccccc}
1 & {{\alpha}^{(a-i+2)}} & {{\alpha}^{2(a-i+2)}} & \cdots & {{\alpha}^{(n-1)(a-i+2)}}\\
\vdots & \vdots & \vdots & \vdots & \vdots\\
1 & {{\alpha}^{(a-1)}} & \cdots & \cdots & {{\alpha}^{(n-1)(a-1)}}\\
1 & {{\alpha}^{a}} & \cdots & \cdots & {{\alpha}^{(n-1)a}}\\
\end{array}
\right]
\end{eqnarray*}
by expanding each entry as a column vector with respect to $\beta$
of $F_{q^2}$. Then it follows that $C_2$ is a code with parameters
${[n, n-2i+3, 2i-2]}_{q}$.

Let $C_1$ be the BCH code of length $n$ over $F_{q}$ generated by
${M}^{(a-i+1)}(x)$ whose parity check matrix $H_{C_1}$ is given by
expanding each entry of the matrix
\begin{eqnarray*}
H_{2, a-i+1} = \\ = \left[
\begin{array}{ccccc}
1 & {{\alpha}^{(a-i+1)}} & {{\alpha}^{2(a-i+1)}} & \cdots & {{\alpha}^{(n-1)(a-i+1)}}\\
\end{array}
\right]
\end{eqnarray*}
with respect to $\beta$, and assume that $C$ is the BCH code
generated by the minimal polynomial ${M}^{(a-i)}(x)$ with parity
check matrix $H_{C}$ given by expanding each entry of the matrix
\begin{eqnarray*}
H_{2, a-i} = \\ = \left[
\begin{array}{ccccc}
1 & {{\alpha}^{(a-i)}} & {{\alpha}^{2(a-i)}} & \cdots & {{\alpha}^{(n-1)(a-i)}}\\
\end{array}
\right]
\end{eqnarray*}
with respect to $\beta$. We know that $C_1$ and $C$ has parameters
${[n, n-2, d\geq 2]}_{q}$.

Rearranging the rows of $H_{C_3}$ we obtain the matrix
\begin{eqnarray*}
H = \\ = \left[
\begin{array}{ccccc}
1 & {{\alpha}^{a}} & \cdots & \cdots & {{\alpha}^{(n-1)a}}\\
1 & {{\alpha}^{(a-1)}} & \cdots & \cdots & {{\alpha}^{(n-1)(a-1)}}\\
\vdots & \vdots & \vdots & \vdots & \vdots\\
1 & {{\alpha}^{(a-i+2)}} & {{\alpha}^{2(a-i+2)}} & \cdots & {{\alpha}^{(n-1)(a-i+2)}}\\
1 & {{\alpha}^{(a-i+1)}} & {{\alpha}^{2(a-i+1)}} & \cdots & {{\alpha}^{(n-1)(a-i+1)}}\\
1 & {{\alpha}^{(a-i)}} & {{\alpha}^{2(a-i)}} & \cdots & {{ \alpha}^{(n-1)(a-i)}} \\
\end{array}
\right].
\end{eqnarray*}
Next we split $H$ into three disjoint submatrices $H_0$ and $H_1$
and $H_2$ (as in Theorem~\ref{main}) of the form

\begin{eqnarray*}
H_0 = \\ = \left[
\begin{array}{ccccc}
1 & {{\alpha}^{a}} & \cdots & \cdots & {{\alpha}^{(n-1)a}}\\
1 & {{\alpha}^{(a-1)}} & \cdots & \cdots & {{\alpha}^{(n-1)(a-1)}}\\
\vdots & \vdots & \vdots & \vdots & \vdots\\
1 & {{\alpha}^{(a-i+2)}} & {{\alpha}^{2(a-i+2)}} & \cdots & {{\alpha}^{(n-1)(a-i+2)}}\\
\end{array}
\right],
\end{eqnarray*}

\begin{eqnarray*}
H_1 = \\ = \left[
\begin{array}{ccccc}
1 & {{\alpha}^{(a-i+1)}} & {{\alpha}^{2(a-i+1)}} & \cdots & {{\alpha}^{(n-1)(a-i+1)}}\\
\end{array}
\right],
\end{eqnarray*}
and
\begin{eqnarray*}
H_2 = \\ = \left[
\begin{array}{ccccc}
1 & {{\alpha}^{(a-i)}} & {{\alpha}^{2(a-i)}} & \cdots & {{\alpha}^{(n-1)(a-i)}}\\
\end{array}
\right],
\end{eqnarray*}
obtaining, in this way, a memory-two convolutional code $V$
generated by the matrix
\begin{eqnarray*}
G(D)=\tilde H_{0}+ \tilde H_1 D + \tilde H_2 D^{2}
\end{eqnarray*}
with parameters $(n, 2i-3, 4; 2, d_{f})_{q}$, where, from Item (c)
of Theorem~\ref{A}, one concludes that $d_{f}\geq d^{\perp}=n-2i$.
The proof is complete.
\end{proof}

Theorem~\ref{mainII} can be easily generalized as one can see in
the next result:


\begin{theorem}\label{mainIII}
Assume that $q=p^{t}$, where $t\geq 2$ is an integer and $p$ is an
odd prime number. Consider that $n=q+1$, $a=\frac{n}{2}$ and let
$r, m$ integers with $r \geq 1$, $m \geq 2$ such that $3 \leq r +
m \leq a - 1$. Then there exist convolutional codes with
parameters $(n, 2r + 1, 2m; m, d_{f}\geq n-2[r+m])_{q}$.
\end{theorem}

\begin{proof}
Let $C$ be the BCH code of length $n$ over $F_{q}$ generated by
the product of the minimal polynomials
\begin{eqnarray*}
C = \langle g (x)  \rangle =\langle {M}^{(a-[r+m])}(x)\cdot
\ldots\cdot {M}^{(a-[r+1])}(x)\cdot\\ \cdot
{M}^{(a-r)}(x)\cdot\ldots\cdot {M}^{(a-1)}(x){M}^{(a)}(x) \rangle.
\end{eqnarray*}
whose parity check matrix $H_{C}$ is obtained from the matrix

\begin{eqnarray*}
H_{2[r+m]+2 , a-[r+m]} = \\ = \left[
\begin{array}{ccccc}
1 & {{\alpha}^{(a-[r+m])}} & {{ \alpha}^{2(a-[r+m])}} & \cdots & {{\alpha}^{(n-1)(a-[r+m])}} \\
\vdots & \vdots & \vdots & \vdots & \vdots\\
1 & {{\alpha}^{(a-[r+1])}} & {{\alpha}^{2(a-[r+1])}} & \cdots & {{\alpha}^{(n-1)(a-[r+1])}}\\
1 & {{\alpha}^{(a-r)}} & {{\alpha}^{2(a-r)}} & \cdots & {{\alpha}^{(n-1)(a-r)}}\\
\vdots & \vdots & \vdots & \vdots & \vdots\\
1 & {{\alpha}^{(a-1)}} & \cdots & \cdots & {{\alpha}^{(n-1)(a-1)}}\\
1 & {{\alpha}^{a}} & \cdots & \cdots & {{\alpha}^{(n-1)a}}\\
\end{array}
\right]
\end{eqnarray*}
by expanding each entry as a column vector over some $F_{q}-$basis
$\beta$ of $F_{q^2}$. We know that $C$ is a MDS code with
parameters ${[n, n-2[r+m]-1, 2[r+m]+2]}_{q}$

We next consider $C_0$ as the BCH code of length $n$ over $F_{q}$
generated by the product of the minimal polynomials
\begin{eqnarray*}
C_0 =\langle g_0 (x)\rangle = \langle {M}^{(a-r)}(x)\cdot \ldots
\cdot {M}^{(a-1)}(x){M}^{(a)}(x)\rangle.
\end{eqnarray*}
whose parity check matrix $H_{C_0}$ is derived from the matrix
\begin{eqnarray*}
H_{2r+2 , a-r} = \\ = \left[
\begin{array}{ccccc}
1 & {{\alpha}^{(a-r)}} & {{\alpha}^{2(a-r)}} & \cdots & {{\alpha}^{(n-1)(a-r)}}\\
\vdots & \vdots & \vdots & \vdots & \vdots\\
1 & {{\alpha}^{(a-1)}} & \cdots & \cdots & {{\alpha}^{(n-1)(a-1)}}\\
1 & {{\alpha}^{a}} & \cdots & \cdots & {{\alpha}^{(n-1)a}}\\
\end{array}
\right]
\end{eqnarray*}
by expanding each entry as a column vector with respect to $\beta$
of $F_{q^2}$. We know that $C_0$ is a MDS code with parameters
${[n, n-2r-1, 2r + 2]}_{q}$.

Let $C_{i}$ for all $1 \leq i \leq m$, be the BCH code of length
$n$ over $F_{q}$ generated by ${M}^{(a-[r+i])}(x)$ whose parity
check matrix $H_{C_{i}}$ is given by expanding each entry of the
matrix
\begin{eqnarray*}
H_{2, a-[r+i]} = \\ = \left[
\begin{array}{ccccc}
1 & {{\alpha}^{(a-[r+i])}} & {{\alpha}^{2(a-[r+i])}} & \cdots & {{\alpha}^{(n-1)(a-[r+i])}}\\
\end{array}
\right]
\end{eqnarray*}
with respect to $\beta$. We know that $C_{i}$ has parameters ${[n,
n-2, d\geq 2]}_{q}$.

Proceeding similarly as in the proof of Theorem~\ref{mainII}, one
obtains a convolutional code $V$ generated by the matrix
\begin{eqnarray*}
G(D)=\tilde H_{0}+ \tilde H_1 D + \tilde H_2 D^{2} + \cdots +
\tilde H_{m} D^{m}
\end{eqnarray*}
with parameters $(n, 2r+1, 2m; m, d_{f})_{q}$, where $d_{f}\geq
n-2[r+m]$.
\end{proof}

\begin{remark}
It is important to observe that the procedure adopted in
Theorem~\ref{mainIII} has several variants and, therefore, several
more new families can be constructed straightforwardly based on
our method.
\end{remark}

\begin{remark}
Unfortunately if one considers $m > 1$, there is no guarantee that
the corresponding convolutional codes are MDS.
\end{remark}

\section{New Quantum MDS-Convolutional codes}\label{V}

As in the classical case, the construction of MDS quantum
convolutional codes is a difficult task. This task is performed in
\cite{Grassl:2005,Grassl:2007,Klapp:2007,Forney:2007} but only in
\cite{Grassl:2005,Klapp:2007} the constructions are made
algebraically. Based on this view point, we propose the
construction of more MDS convolutional stabilizer codes.

It is well known that convolutional stabilizer codes can be
constructed from classical convolutional codes ( see for example
\cite[Proposition 1 and 2]{Aly:2007}). In the first construction,
one utilizes convolutional codes endowed with the Euclidean inner
product and in the second one, the codes are endowed with the
Hermitian inner product. Considering the $q$-ary cosets modulo
$n=q+1$ as given in the previous section, it is easy to see that
the dual-containing property with respect to the Euclidean inner
product does not hold for (classical) convolutional codes derived
from block codes with defining set of this type. However, when
considering cyclic codes endowed with the Hermitian inner product
one can show the existence of convolutional codes, derived from
them, which are (Hermitian) self-orthogonal (see
Lemma~\ref{cycloH}). This fact permits the construction of MDS
quantum convolutional codes (in the sense that they attain the
generalized quantum Singleton bound (Theorem~\ref{SingC}) as it is
shown in Theorem~\ref{main1}, given in the following. More
precisely, we utilize the MDS-convolutional codes constructed in
the previous section for constructing quantum MDS convolutional
codes. Before proceeding further, we need the following result:

\begin{lemma}\label{cycloH}
Assume $q=2^t$, where $t$ is an integer such that $t\geq 1$,
$n=q^2+1$ and let $a=\frac{q^2}{2}$. If $C$ is the cyclic code
whose defining set $Z$ is given by $Z = {\mathcal{C}}_{a-i} \cup
\ldots \cup {\mathcal{C}}_{a}$, where $0\leq i \leq \frac{q}{2} -
1$, then $C$ is Hermitian dual-containing.
\end{lemma}
\begin{proof}
See \cite[Lemma 4.2]{LaGuardia:2011}.
\end{proof}

Although Theorem~\ref{main1} is a Corollary of Theorem~\ref{main},
we consider it as a theorem because the resulting quantum
convolutional codes are MDS.

\begin{theorem}\label{main1}
Assume $q=2^t$, where $t\geq 3$ is an integer, $n=q^2+1$ and
consider that $a=\frac{q^2}{2}$. Then there exist quantum MDS
convolutional codes with parameters $[(n, n-4i, 1; 2, 2i+3)]_{q}$,
where $2\leq i \leq \frac{q}{2} - 2$.
\end{theorem}
\begin{proof}
We consider the same notation utilized in Theorem~\ref{main}. We
know that $\gcd (n, q^{2}) = 1$. From Theorem~\ref{main}, there
exists a classical convolutional MDS code with parameters $(n,
n-2i, 2; 1, 2i+3)_{q^2}$, for each $2\leq i \leq \frac{q}{2} - 2$.
This code is the Euclidean dual $V^{\perp}$ of the convolutional
code $V$ whose parameters are given by $(n, 2i, 2; 1, d_f)_{q^2}$.
The codes $V^{\perp}$ and $V^{{\perp}_{h}}$ have the same degree
as code (see the proof of Theorem 7 in \cite{Klapp:2007}).
Additionally, it is straightforward to check that
wt$(V^{\perp})$=wt$(V^{{\perp}_{h}})$, so $V^{{\perp}_{h}}$ has
parameters $(n, n-2i, 2; m^{*} , 2i+3)_{q^2}$. From
Lemma~\ref{cycloH} and from Theorem~\ref{A} Item (b), one has $V
\subset V^{{\perp}_{h}}$. Applying Lemma~\ref{BB}, there exists an
${[(n, n-4i, 1; 2, d_f\geq 2i+3)]}_{q}$ convolutional stabilizer
code, for each $2\leq i \leq \frac{q}{2} - 2$. Replacing the
parameters of the previously constructed codes in the quantum
generalized Singleton bound (Theorem~\ref{SingC}) one has the
equality $2i+3 = 2i\left( \left\lfloor \frac{4}{2n-4i}
\right\rfloor + 1\right) + 2 + 1$. Therefore, there exist
MDS-convolutional stabilizer codes with parameters ${[(n, n-4i, 1;
2, 2i+3)]}_{q}$, for each $2\leq i \leq \frac{q}{2} - 2$.
\end{proof}

\begin{example}
To illustrate the previous construction, assume that $q=8$, $n=65$
and $i=2$. Applying Theorem~\ref{main1} there exists an $[(65, 57,
1; 2 , 7)]_{8}$ convolutional stabilizer code that attains the
generalized quantum Singleton bound.

Considering $q=16$, $n=257$ and $i=2, 3, 4, 5$ then one has
quantum MDS codes with parameters $[(257, 249, 1; 2, 7)]_{16}$,
$[(257, 245, 1; 2, 9)]_{16}$, $[(257, 241, 1; 2, 11)]_{16}$,
$[(257, 237, 1; 2, 13)]_{16}$, respectively, and so on.
\end{example}


\section{Summary}\label{VI}
In this paper we have constructed several new families of
multi-memory classical convolutional BCH codes. The families of
unit-memory codes are optimal in the sense that they attain the
classical generalized Singleton bound. Moreover, we also have
constructed families of unit-memory optimal quantum convolutional
codes in the sense that these codes attain the quantum generalized
Singleton bound. All the constructions presented here are
performed algebraically and not by exhaustively computational
search.

\section*{Acknowledgment}
I would like to thank the anonymous referees and the Associate
Editor Prof. Dr. Alexei Ashikhmin for their valuable comments and
suggestions that improve significantly the quality and the
presentation of this paper. This research has been partially
supported by the Brazilian Agencies CAPES and CNPq.

\textbf{Giuliano G. La Guardia received the M.S. degree in pure
mathematics in 1998 and the Ph.D. degree in electrical engineering
in 2008, both from the State University of Campinas (UNICAMP), São
Paulo, Brazil. Since 1999, he has been with the Department of
Mathematics and Statistics, State University of Ponta Grossa,
where he is an Associate Professor. His research areas include
theory of classical and quantum codes, matroid theory, and error
analysis.}


\begin{thebibliography}{10}

\bibitem{Aido:2004}
A. C. A. de Almeida and R. Palazzo Jr..
\newblock A concatenated $[(4, 1, 3)]$ quantum convolutional code.
\newblock In {\em Proc. IEEE Inform. Theory Workshop (ITW)}, pp. 28--33, 2004.

\bibitem{Aly:2007}
S. A. Aly, M. Grassl, A. Klappenecker, M. R\"otteler, P. K. Sarvepalli.
\newblock Quantum convolutional BCH codes.
\newblock e-print arXiv:quant-ph/0703113.

\bibitem{Klapp:2007}
S. A. Aly, A. Klappenecker, P. K. Sarvepalli.
\newblock Quantum convolutional codes derived from Reed-Solomon and Reed-Muller codes.
\newblock e-print arXiv:quant-ph/0701037.

\bibitem{Salah:2007}
S. A. Aly, A. Klappenecker, and P. K. Sarvepalli.
\newblock On quantum and classical BCH codes.
\newblock {\em IEEE Trans. Inform. Theory}, 53(3):1183--1188, March 2007.

\bibitem{Ashikmin:2001}
A. Ashikhmin and E. Knill.
\newblock Non-binary quantum stabilizer codes.
\newblock {\em IEEE Trans. Inform. Theory}, 47(7):3065--3072, November 2001.

\bibitem{Bierbrauer:2000}
J. Bierbrauer and Y. Edel.
\newblock Quantum twisted codes.
\newblock {\em J. Comb. Designs}, 8:174--188, 2000.

\bibitem{Calderbank:1998}
A. R. Calderbank, E. M. Rains, P. W. Shor, and N. J. A. Sloane.
\newblock Quantum error correction via codes over $GF(4)$.
\newblock {\em IEEE Trans. Inform. Theory}, 44(4):1369--1387, July 1998.

\bibitem{Chen:2005}
H. Chen, S. Ling, and C. P. Xing.
\newblock Quantum codes from concatenated algebraic geometric codes.
\newblock {\em IEEE. Trans. Inform. Theory}, 51(8):2915--2920, august 2005.

\bibitem{Cohen:1999}
G. D. Cohen, S. B. Encheva, and S. Litsyn.
\newblock On binary constructions of quantum codes.
\newblock {\em IEEE Trans. Inform. Theory}, 45(7):2495--2498, July 1999.


\bibitem{Climent:2008}
J. J. Climent, V. Herranz, C. Perea.
\newblock Linear system modelization of concatenated block and convolutional codes.
\newblock {\em Linear Algebra and its Applications}, 429(5-6):1191--1212, 2008.

\bibitem{Forney:1970}
G. D. Forney Jr.
\newblock Convolutional codes I: algebraic structure.
\newblock {\em IEEE Trans. Inform. Theory}, 16(6):720--738, November 1970.

\bibitem{Forney:2007}
G. D. Forney Jr., M. Grassl, S. Guha.
\newblock Convolutional and tail-biting quantum error-correcting codes.
\newblock {\em IEEE Trans. Inform. Theory}, 53(3):865–-880, March 2007.

\bibitem{Grassl:2003}
M. Grassl, T. Beth, and M. R\"otteler.
\newblock On optimal quantum codes.
\newblock {\em Int. J. Quantum Inform.}, 2(1):757--766, 2004.

\bibitem{Grassl:2005}
M. Grassl and M. R\"{o}tteler.
\newblock Quantum block and convolutional codes from self-orthogonal product codes.
\newblock In {\em Proc. Int. Symp. Inform. Theory (ISIT)}, pp. 1018--1022, 2005.

\bibitem{Grassl:2006}
M. Grassl and M. R\"otteler.
\newblock Non-catastrophic encoders and encoder inverses for quantum convolutional codes.
\newblock In {\em Proc. Int. Symp. Inform. Theory (ISIT)}, pp. 1109–-1113, 2006.

\bibitem{Grassl:2007}
M. Grassl and M. R\"otteler.
\newblock Constructions of quantum convolutional codes.
\newblock e-print arXiv:quant-ph/0703182.

\bibitem{Hamada:2008}
M. Hamada.
\newblock Concatenated quantum codes constructible in polynomial time:
efficient decoding and error correction.
\newblock {\em IEEE Trans. Inform. Theory}, 54(12):5689--5704, December 2008.

\bibitem{Hole:2000}
K. J. Hole.
\newblock On classes of convolutional codes that are not asymptotically catastrophic.
\newblock {\em IEEE Trans. Inform. Theory}, 46(2):663--669, March 2000.

\bibitem{Huffman:2003}
W. C. Huffman and V. Pless.
\newblock {\em Fundamentals of Error-Correcting Codes.}
\newblock University Press, Cambridge, 2003.

\bibitem{Johannesson:1999}
R. Johannesson and K. S. Zigangirov.
\newblock {\em Fundamentals of Convolutional Coding.}
\newblock Digital and Mobile Communication, Wiley-IEEE Press, 1999.

\bibitem{Ketkar:2006}
A. Ketkar, A. Klappenecker, S. Kumar, and P. K. Sarvepalli.
\newblock Nonbinary stabilizer codes over finite fields.
\newblock {\em IEEE Trans. Inform. Theory}, 52(11):4892--4914, November 2006.

\bibitem{LaGuardia:2009}
G. G. La Guardia.
\newblock Constructions of new families of nonbinary quantum codes.
\newblock {\em Phys. Rev. A}, 80(4):042331(1--11), October 2009.

\bibitem{LaGuardia:2011}
G. G. La Guardia.
\newblock New quantum MDS codes.
\newblock {\em IEEE Trans. Inform. Theory}, 57(8):5551--5554, August 2011.

\bibitem{LaGuardia:2010}
G. G. La Guardia and R. Palazzo Jr..
\newblock Constructions of new families of nonbinary CSS codes.
\newblock {\em Discrete Math.}, 310(21):2935--2945, November 2010.


\bibitem{Lee:1976}
L. N. Lee.
\newblock Short unit-memory byte-oriented binary convolutional codes having maximum free distance.
\newblock {\em IEEE Trans. Inform. Theory}, 22:349--352, May 1976.


\bibitem{Gluesing:2006}
H. Gluesing-Luerssen, J. Rosenthal and R. Smarandache.
\newblock Strongly MDS convolutional codes.
\newblock {\em IEEE Trans. Inform. Theory}, 52:584--598, 2006.


\bibitem{Schmale:2006}
H. Gluesing-Luerssen, W. Schmale.
\newblock Distance bounds for convolutional codes and some optimal codes.
\newblock e-print arXiv:math/0305135.


\bibitem{Luerssen:2006}
H. Gluesing-Luerssen and W. Schmale.
\newblock  On doubly-cyclic convolutional codes.
\newblock {\em Applicable Algebra in Eng. Comm. Comput.}, 17(2):151--170, 2006.

\bibitem{Luerssen:2008}
H. Gluesing-Luerssen and F-L Tsang.
\newblock A matrix ring description for cyclic convolutional codes.
\newblock {\em Advances in Math. Communications}, 2(1):55--81, 2008.


\bibitem{Hutchinson:2005}
R. Hutchinson, J. Rosenthal and R. Smarandache.
\newblock Convolutional codes with maximum distance profile.
\newblock {\em Systems and Control Letters}, 54(1):53--63, 2005.


\bibitem{Iglesias:2009}
J. I. Iglesias-Curto.
\newblock Generalized AG Convolutional Codes.
\newblock {\em Advances in Mathematics of Communications}, 3(4):317--328,
2009.


\bibitem{Macwilliams:1977}
F. J. MacWilliams and N. J. A. Sloane.
\newblock {\em The Theory of Error-Correcting Codes}.
\newblock North-Holland, 1977.

\bibitem{Nielsen:2000}
M. A. Nielsen and I. L. Chuang.
\newblock {\em Quantum Computation and Quantum Information}.
\newblock Cambridge University Press, 2000.

\bibitem{Ollivier:2003}
H. Ollivier and J.-P. Tillich.
\newblock Description of a quantum convolutional code.
\newblock {\em Phys. Rev. Lett.}, 91(17):1779021–4, 2003.


\bibitem{Ollivier:2004}
H. Ollivier and J.-P. Tillich.
\newblock Quantum convolutional codes: fundamentals.
\newblock e-print arXiv:quant-ph/0401134.


\bibitem{Piret:1988}
Ph. Piret.
\newblock {\em Convolutional Codes: An Algebraic Approach.}
\newblock Cambridge, Massachusetts: The MIT Press, 1988.

\bibitem{Piret:1988art}
Ph. Piret.
\newblock A convolutional equivalent to Reed-Solomon codes.
\newblock {\em Philips J. Res.}, 43:441--458, 1988.


\bibitem{Rosenthal:1999}
J. Rosenthal and R. Smarandache.
\newblock Maximum distance separable convolutional codes.
\newblock {\em Applicable Algebra in Eng. Comm. Comput.}, 10:15--32, 1998.

\bibitem{York:1999}
J. Rosenthal and E. V. York.
\newblock BCH convolutional codes.
\newblock {\em IEEE Trans. Inform. Theory}, 45(6):1833–-1844, 1999.

\bibitem{Rosenthal:2001}
R. Smarandache, H. G.-Luerssen, J. Rosenthal.
\newblock Constructions of MDS-convolutional codes.
\newblock {\em IEEE Trans. Inform. Theory}, 47(5):2045--2049, July 2001.

\bibitem{SK:2005}
P. K. Sarvepalli and A. Klappenecker.
\newblock Nonbinary quantum Reed-Muller codes.
\newblock In {\em Proc. Int. Symp. Inf. Theory (ISIT)}, pp. 1023--1027, 2005.

\bibitem{Steane:1999}
A. Steane.
\newblock Enlargement of Calderbank-Shor-Steane quantum codes.
\newblock {\em IEEE Trans. Inform. Theory}, 45(7):2492--2495, November 1999.


\bibitem{Tan:2010}
P. Tan and J. Li.
\newblock Efficient Quantum Stabilizer Codes: LDPC and LDPC-Convolutional Constructions.
\newblock {\em IEEE Trans. Inform. Theory}, 56(1):476--491, January 2010.


\bibitem{Wilde:2008}
M. M. Wilde and T. A. Brun.
\newblock Unified quantum convolutional coding.
\newblock In {\em Proc. Int. Symp. Inf. Theory (ISIT)}, pp. 359--363, 2008.


\bibitem{Wilde:2010}
M. M. Wilde and T. A. Brun.
\newblock Entanglement-assisted quantum convolutional coding.
\newblock {\em Phys. Rev. A}, 81(4):042333, April 2010.

\end{thebibliography}
\end{document}